\RequirePackage{ifpdf}
\ifpdf 
\documentclass[pdftex]{sigma}
\else
\documentclass{sigma}
\fi

\begin{document}

	
\renewcommand{\PaperNumber}{114}

\FirstPageHeading

\renewcommand{\thefootnote}{$\star$}

\ShortArticleName{Some Sharp $L^{2}$ Inequalities for Dirac Type
Operators}

\ArticleName{Some Sharp $\boldsymbol{L^{2}}$ Inequalities for Dirac Type
Operators\footnote{This paper is a
contribution to the Proceedings of the 2007 Midwest
Geometry Conference in honor of Thomas~P.\ Branson. The full collection is available at
\href{http://www.emis.de/journals/SIGMA/MGC2007.html}{http://www.emis.de/journals/SIGMA/MGC2007.html}}}

\Author{Alexander BALINSKY~$^\dag$ and John RYAN~$^\ddag$}

\AuthorNameForHeading{A. Balinsky and J. Ryan}

\Address{$^\dag$~Cardif\/f School of Mathematics, Cardif\/f University,\\
$\phantom{^\dag}$~Senghennydd Road, Cardif\/f, CF 24 4AG, UK}
\EmailD{\href{mailto:BalinskyA@cardiff.ac.uk}{BalinskyA@cardiff.ac.uk}}
\URLaddressD{\url{http://www.cf.ac.uk/maths/people/balinsky.html}}

\Address{$^\ddag$~Department of Mathematics, University of
Arkansas, Fayetteville, AR 72701, USA}
\EmailD{\href{mailto:jryan@uark.edu}{jryan@uark.edu}}
\URLaddressD{\url{http://comp.uark.edu/~jryan/}}

\ArticleDates{Received August 31, 2007, in f\/inal form November
14, 2007; Published online November 25, 2007}

\Abstract{We use the spectra of Dirac type operators on the sphere
$S^{n}$ to produce sharp $L^{2}$ inequalities on the sphere. These
operators include the Dirac operator on $S^{n}$, the conformal
Laplacian and Paenitz operator. We use the Cayley transform, or
stereographic projection, to obtain similar inequalities for
powers of the Dirac operator  and their inverses in
$\mathbb{R}^{n}$.}


\Classification{15A66; 26D10; 34L40}

\rightline{\it This paper is dedicated to the memory of Tom Branson}

\section{Introduction}

Sobolev and Hardy type inequalities play an important role in many
areas of mathematics and mathematical physics. They have become
standard tools in existence and regularity theories for solutions
to partial dif\/ferential equations, in calculus of variations, in
geometric measure theory and in stability of matter.
In analysis a number of inequalities like the
Hardy--Littlewood--Sobolev inequality in $\mathbb{R}^{n}$ are
obtained by f\/irst obtaining these inequalities on the compact
manifold $S^{n}$ and then using stereographic projections to
$\mathbb{R}^{n}$ to obtain the analogous sharp inequality in that
setting. See for instance \cite{ll}. This technique is also used
in mathematical physics to obtain zero modes of Dirac equations in
$\mathbb{R}^{3}$ (see \cite{e}).

In fact the stereographic projection corresponds to the Cayley
transformation from $S^{n}$ minus the north pole to Euclidean
space. Here we shall use this Cayley transformation to obtain some
sharp $L^{2}$ inequalities on the sphere for a family of Dirac
type operators. The main trick here is to employ a lowest
eigenvalue for these operators and then use intertwining operators
for the Dirac type operators to obtain analogous sharp
inequalities in $\mathbb{R}^{n}$.

Our eventual hope is to extend the results presented here to
obtain suitable $L^{p}$ inequalities for the Dirac type operators
appearing here, particularly the Dirac operator on
$\mathbb{R}^{n}$.

\section{Preliminaries}

We shall consider $\mathbb{R}^{n}$ as embedded in the real,
$2^{n}$ dimensional Clif\/ford algebra $Cl_{n}$ so that for each
$x\in\mathbb{R}^{n}$ we have $x^{2}=-\|x\|^{2}$. Consequently if
$e_{1},\ldots, e_{n}$ is an orthonormal basis for $\mathbb{R}^{n}$
then
\[e_{i}e_{j}+e_{j}e_{i}=-2\delta_{ij}\]
and
\[1,e_{1},\ldots, e_{n},e_{1}e_{2},\ldots, e_{n-1}e_{n},\ldots, e_{j_{1}},
\ldots, e_{j_{r}},\ldots,e_{1},\ldots, e_{n}\]
is an orthonormal basis for $Cl_{n}$, with $1\leq r\leq n$ and $j_{1}<\cdots <j_{r}$.

Note that for each $x\in \mathbb{R}^{n}\backslash\{0\}$ we have
that $x$ is invertible, with multiplicative inverse
$\frac{-x}{\|x\|^{2}}$. Here, up to a sign, $x^{-1}$ is the Kelvin
inverse of $x$. It follows that $\{A\in Cl_{n}:A=x_{1}\cdots
x_{m}$ with $m\in\mathbb{N}$ and $x_{1},\ldots,
x_{m}\in\mathbb{R}^{n}\backslash\{0\}\}$ is a subgroup of
$Cl_{n}$. We shall denote this group by $GPin(n)$.

We shall need the following anti-automorphisms on $Cl_{n}$:
\[\sim:Cl_{n}\rightarrow Cl_{n}:e_{j_{1}}\cdots e_{j_{r}}\rightarrow e_{j_{r}}\cdots e_{j_{1}}\]
and
\[-:Cl_{n}\rightarrow Cl_{n}:e_{j_{1}}\cdots e_{j_{r}}\rightarrow (-1)^{r}e_{j_{r}}\cdots e_{j_{1}}.\]

For $A\in Cl_{n}$ we denote $\sim(A)$ by $\tilde{A}$ and we denote $-(A)$ by $\overline{A}$.
Note that for $A=a_{0}+\cdots +a_{1\dots n}e_{1}\cdots e_{n}$ the scalar part of
$A\overline{A}$ is $a_{0}^{2}+\cdots+a_{1\dots n}^{2}:=\|A\|^{2}$.

\begin{lemma}\label{lem1}
If $A\in GPin(n)$ and $B\in Cl_{n}$ then $\|AB\|=\|A\|\|B\|$.
\end{lemma}

\begin{proof} $\overline{AB}AB=\overline{B}$ $
 \overline{A}AB=\overline{B}\|A\|^{2}B=\|A\|^{2}\overline{B}B$.
Therefore
$Sc(\overline{AB}AB)=\|A\|^{2}Sc(\overline{B}B)=\|A\|^{2}\|B\|^{2}$,
where $Sc(C)$ is the scalar part of $C$ for any $C\in Cl_{n}$. The
result follows.
\end{proof}

In \cite{ahlfors} it is shown that if $y=M(x)$ is a M\"{o}bius
transformation then $M(x)=(ax+b)(cx+d)^{-1}$ where $a$, $b$, $c$
and $d\in Cl_{n}$ and satisfy the conditions

\begin{enumerate}\itemsep=0pt

\item[(i)] $a, b, c, d\in GPin(n)$.

\item[(ii)] $a\tilde{c}, \tilde{c}d, \tilde{d}b, \tilde{b}a\in\mathbb{R}^{n}$

\item[(iii)] $a\tilde{d}-c\tilde{c}\in\mathbb{R}\backslash\{0\}$.
\end{enumerate}

In particular if we regard $\mathbb{R}^{n}$ as embedded in
$\mathbb{R}^{n+1}$ in the usual way, then
$y=(e_{n+1}x+1)(x+e_{n+1})^{-1}$ is the Cayley transformation from
$\mathbb{R}^{n}$ to the unit sphere $S^{n}$ in $\mathbb{R}^{n+1}$.
This map corresponds to the stereographic projection of
$\mathbb{R}^{n}$ onto $S^{n}\backslash\{e_{n+1}\}$.

The Dirac operator in $\mathbb{R}^{n}$ is
$\sum_{j=1}^{n}e_{j}\frac{\partial}{\partial x_{j}}$. Note that
$D^{2}=-\triangle_{n}$, where $\triangle_{n}$ is the Laplacian in~$\mathbb{R}^{n}$, and $D^{4}$ is the bi-Laplacian
$\triangle_{n}^{2}$.

\section[Eigenvectors of the Dirac-Beltrami operator on $S^{n}$]{Eigenvectors of the Dirac--Beltrami operator on $\boldsymbol{S^{n}}$}

We start with the Dirac operator
$D_{n+1}=\sum_{j=1}^{n+1}e_{j}\frac{\partial}{\partial x_{j}}$
in $\mathbb{R}^{n+1}$. For each point in
$x\in\mathbb{R}^{n+1}\backslash\{0\}$ this operator can be
rewritten as $x^{-1}xD_{n+1}$. Now $xD_{n+1}=x\wedge
D_{n+1}-x\cdot D_{n+1}$. Now $x\wedge D_{n+1}=\sum_{1\leq
j<k\leq n+1}e_{i}e_{j}(x_{j}\frac{\partial}{\partial
x_{k}}-x_{k}\frac{\partial}{\partial x_{j}})$ and $x \cdot D_{n+1}$ is
the Euler operator
$\sum_{j=1}^{n+1}x_{j}\frac{\partial}{\partial
x_{j}}=r\frac{\partial}{\partial r}$ where \mbox{$r=\|x\|$}. It is well
known and easily verif\/ied fact that if $p_{m}(x)$ is a polynomial
homogeneous of degree $m\in\mathbb{N}$ then $x\cdot
Dp_{m}(x)=mp_{m}(x)$. So in particular if $D_{n+1}p_{m}(x)=0$ then
$x\wedge D_{n+1}p_{m}(x)=-mp_{m}(x)$. So $p_{m}(x)$ is an
eigenvector of the operator $x\wedge D_{n+1}$.

Further it is also easily verif\/ied that if $q_{m}(x)$ is
homogeneous of degree $m\in -\mathbb{N}$ then \mbox{$x\cdot
D_{n+1}q_{m}(x)=mq_{m}(x)$}. So if $D_{n+1}q=0$ then $x\wedge
D_{n+1}q=mq$ and $q$ is an eigenvector of $x\wedge D_{n+1}$.

Now let us suppose that $p_{m}:\mathbb{R}^{n+1}\rightarrow
Cl_{n+1}$ is a harmonic polynomial homogeneous of degree
$m\in\mathbb{N}$. In \cite{ryan} it is shown that
$p_{m}(x)=p_{m,1}(x)+xp_{m-1,2}(x)$ where
$D_{n+1}p_{m,1}(x)=D_{n+1}p_{m-1,2}(x)=0$, with $p_{m,1}(x)$
homogeneous of degree $m$ and $p_{m-1,2}(x)$ homogeneous of degree
$m-1$.

\begin{definition}
Suppose $U$ is a domain in $\mathbb{R}^{n+1}$ and
$f:U\rightarrow Cl_{n+1}$ is a $C^{1}$ function satisfying
$D_{n+1}f=0$ then $f$ is called a left monogenic function.
\end{definition}

A similar def\/inition can be given for right monogenic functions. See \cite{bds} for details.

In \cite{sommen} it is shown that if $U$ is a domain in
$\mathbb{R}^{n+1}\backslash\{0\}$ and $f:U\rightarrow Cl_{n+1}$ is
left monogenic then the function $G(x)f(x^{-1})$ is left monogenic
on the domain $U^{-1}=\{x\in\mathbb{R}^{n+1}:x^{-1}\in U\}$ where
$G(x)=\frac{x}{\|x\|^{n+1}}$. Note that on $S^{n}\cap U^{-1}$ for
any function $g$ def\/ined on $U$ the functions $G(x)g(x^{-1})$ and
$xg(x^{-1})$ coincide.

Let ${\it{H}}_{m}$ denote the restriction to $S^{n}$ of the space
of $Cl_{n}$ valued harmonic polynomials homogeneous of degree
$m\in\mathbb{N}\cup\{0\}$. This is the space of spherical
harmonics homogeneous of degree~$m$. Further let ${\it{P}}_{m}$
denote the restriction to $S^{n}$ of left  monogenic polynomials
homogeneous of degree $m\in\mathbb{N}\cup\{0\}$, and let
${\it{Q}}_{m}$ denote the restriction to $S^{n}$ of the space of
left monogenic functions homogeneous of degree $-n-m$ where
$m=0,1,2,\ldots$. Then we have illustrated that
${\it{H}}_{m}={\it{P}}_{m}\oplus {\it{Q}}_{m}$. This result was
established in the quaternionic case in \cite{sudbery} and
independently for all $n$ in \cite{sommen}.

As $L^{2}(S^{n})=\sum_{m=0}^{\infty}{\it{H}}_{m}$ then it
follows that $L^{2}(S^{n})=\sum_{m=0}^{\infty}{\it{P}}_{m}\oplus
{\it{Q}}_{m}$ where $L^{2}(S^{n})$ is the space of $Cl_{n+1}$
valued square integrable functions on $S^{n}$. Further we have
shown that if $p_{m}\in{\it{P}}_{m}$ then $p_{m}$ is an
eigenvector of the Dirac--Beltrami operator $\Gamma_{w}$, where
$\Gamma_{w}$ is the restriction to $S^{n}$ of $x\wedge D_{n+1}$.
Here $w\in S^{n}$. Further $p_{m}$ has eigenvalue $m$. Also if
$q_{m}\in{\it{Q}}_{m}$ is an eigenvector of~$\Gamma_{w}$ with
eigenvalue $-n-m$. Consequently the spectrum,
$\sigma(\Gamma_{w})$ of the Dirac--Beltrami operator $\Gamma_{w}$
is $\{0\}\cup\mathbb{N}\cup\{-n,-n-1,\ldots\}$. As
$0\in\sigma(\Gamma_{w})$ the linear operator
$\Gamma_{w}:L^{2}(S^{n})\rightarrow L^{2}(S^{n})$ is not
invertible.

Further within our calculations we have also shown that if
$h:S^{n}\rightarrow Cl_{n+1}$ is a $C^{1}$ function then
$\Gamma_{w}wh(w)=-nwh(w)-w\Gamma_{w}h(w)$. By completeness this
extends to all of $L^{2}(S^{n})$.

\section[Dirac type operators in $\mathbb{R}^{n}$ and $S^{n}$ and conformal structure]{Dirac type operators in $\boldsymbol{\mathbb{R}^{n}}$ and $\boldsymbol{S^{n}}$ and conformal structure}

The Dirac type operators that we shall consider here in
$\mathbb{R}^{n}$ are integer powers of $D$. Namely $D^{m}$ for
$m\in\mathbb{N}$. In \cite{bojarski} it is shown that if
$y=M(x)=(ax+b)(cx+d)^{-1}$ is a M\"{o}bius transformation and
$f:U\rightarrow Cl_{n}$ is a $C^{k}$ function then
$D^{k}J_{k}(M,x)f(M(x))=J_{-k}(M,x)D^{k}f(y)$, where
$J_{m}(M,x)=\frac{\widetilde{cx+d}}{\|cx+d\|^{n+m}}$ for $m$ an
odd integer and $J_{m}(M,x)=\frac{1}{\|cx+d\|^{n+m}}$ for $m$ an
even integer. This describes intertwining operators for powers of
the Dirac operator in $\mathbb{R}^{n}$ under actions of the
conformal group.

In \cite{ryan2} the Cayley transformation
$C(x)=(e_{n+1}x+1)(x+e_{n+1})^{-1}$ is used to transform the
euclidean Dirac operator, $D$, to a Dirac operator, $D_{S}$, over
$S^{n}$. This Dirac operator is also described in
\cite{beckner,bo} and elsewhere. In \cite{cm} a simple geometric
argument is used to show that $D_{S}=w(\Gamma_{w}+\frac{n}{2})$.
Using the spectrum of $\Gamma_{w}$ it can be seen that on
$L^{2}(S^{n})$ the operator $D_{S}$ has spectrum
$\sigma(D_{S})=\sigma(\Gamma_{w})+\frac{n}{2}$ which is always
non-zero. In fact $\sigma(D_{S})=\{\frac{n}{2}+m:m=0,1,
2,\ldots\}\cup\{-\frac{n}{2}-m:m=0,1,2, 3,\ldots\}$. Consequently
$D_{S}$ has an inverse $D_{S}^{-1}$ on $L^{2}(S^{n})$ and
following \cite{c} the spherical Dirac operator has as fundamental
solution $C_{1}(w,y):=D_{S}^{-1}\star\delta_{y}$ for each $y\in
S^{n}$. Here $\delta_{y}$ is the Dirac delta function. In
\cite{ryan2} it is shown that
$C_{1}(w,y)=\frac{1}{\omega_{n}}\frac{y-w}{\|y-w\|^{n}}$ where
$\omega_{n}$ is the surface area of the unit sphere in
$\mathbb{R}^{n}$. See also \cite{lr}.

In fact one can for each $\alpha\in\mathbb{C}$ introduce the Dirac
operator $D_{\alpha}:=w(\Gamma+\alpha)$. Provided $-\alpha$ is not
in $\sigma(\Gamma_{w})$ then $D_{\alpha}$ is invertible and has
fundamental solution $D_{\alpha}^{-1}\star\delta_{y}$. See
\cite{vl} for further details. A main advantage that the Dirac
operator $D_{S}$ has over $D_{\alpha}$ for $\alpha$ not equal to
$\frac{n}{2}$ is that~$D_{S}$ is conformally invariant. We shall
use this fact to obtain our sharp inequalities in
$\mathbb{R}^{n}$.

By applying $D_{S}$ to
$C_{2}(w,y):=\frac{1}{(n-2)\omega_{n}}\frac{1}{\|w-y\|^{n-2}}$ it
may be determined \cite{lr} that
$D_{S}C_{2}(w,y)=C_{1}(w,y)-wC_{2}(w,y)$. Consequently
$D_{S}(D_{S}-w)C_{2}(w,y)=\delta_{y}$.

It is well known that in $\mathbb{R}^{n+1}$ the Laplacian in
spherical co-ordinates is
\[\frac{\partial^{2}}{\partial r^{2}}+\frac{n}{r}\frac{\partial}{\partial r}+\frac{1}{r^{2}}\triangle_{w},\]
where $\triangle_{w}$ is the Laplace--Beltrami operator on $S^{n}$.
It follows from arguments presented in \cite{sudbery} that
$\triangle_{w}=((1-n)-\Gamma_{w})\Gamma_{w}$.
Using this fact we can now simplify the expression $D_{S}(D_{S}-w)$ as follows:
\[D_{S}(D_{S}-w)=D_{S}^{2}-D_{S}w.\]
But
\begin{gather*}
D_{S}w=w\left(\Gamma_{w}+\frac{n}{2}\right)w=w^{2}\left(-\Gamma_{w}-n+\frac{n}{2}\right)=-wD_{S}.
\end{gather*}
So
\begin{gather*}
D_{S}^{2}-D_{S}w=D_{S}^{2}+wD_{S} =D_{S}w\left(\Gamma_{w}+\frac{n}{2}\right)+wD_{S}
=-wD_{S}\left(\Gamma_{w}+\frac{n}{2}\right)+wD_{S}\\
\phantom{D_{S}^{2}-D_{S}w}{} =\left(\Gamma_{w}+\frac{n}{2}\right)\left(\Gamma_{w}+\frac{n}{2}\right)-\left(\Gamma_{w}+\frac{n}{2}\right)
=\Gamma_{w}^{2}+n\Gamma_{w}-\Gamma_{w}+\frac{n^{2}}{4}-\frac{n}{2}\\
\phantom{D_{S}^{2}-D_{S}w}{}=-\triangle_{w}+\frac{n^{2}-2n}{4}
=-\triangle_{w}+\frac{n}{2}\left(\frac{n-2}{2}\right).
\end{gather*}
This operator is the conformal Laplacian $\triangle_{S}$ on $S^{n}$
described in \cite{beckner,bo} and elsewhere.

One may also introduce generalized spherical Laplacians of the
type
$\triangle_{\alpha,\beta}=(\Gamma_{w}+\alpha)(\Gamma_{w}+\beta)$
where $\alpha$ and $\beta\in\mathbb{C}$. Provided $-\alpha$ and
$-\beta$ do not belong to $\sigma(\Gamma_{w})$ then the Laplacian
is invertible with fundamental solution
$\triangle_{\alpha,\beta}^{-1}\star\delta_{y}$. In \cite{lr} it is
shown that $\triangle_{\alpha,-\alpha-n+1}$ is a scalar valued
operator. This operator is invertible provided $\alpha$ does not
belong to $\sigma(\Gamma_{w})$. Further, explicit formulas for
this operator are presented in \cite{lr}.

Again a main advantage of the conformal Laplacian, $\triangle_{S}$
over the other choices of Laplacians presented here is its
conformal covariance. We shall see the advantage of this in the
next section.

In \cite{lr} we introduce the operators
\[D_{S}^{(k)}:=D_{S}(D_{S}-w)\cdots\left(D_{S}-\frac{(k-1)}{2}w\right)\]
for $k$ odd, $k>0$, and
\[D_{S}^{(k)}:=D_{S}(D_{S}-w)(D_{S}-w)\cdots\left(D_{S}-\frac{k}{2}w\right)\]
for $k$ even and $k>0$.

When $k=1$ we obtain $D_{S}$, when $k=2$ we obtain $\triangle_{S}$
and when $k=4$ the operator
$D_{S}^{(4)}=\triangle_{S}(D_{S}-w)(D_{S}-2w)$. Moreover
\[(D_{S}-w)(D_{S}-2w)=D_{S}^{2}-wD_{S}-2D_{S}w-2 =D_{S}^{2}+wD_{S}-2=-\triangle_{S}-2.\]
Consequently $D_{S}^{(4)}=-\triangle_{S}(\triangle_{S}+2)$.
When $n=4$ this operator becomes $-\triangle_{S}(\triangle_{S}+2)$ is the
Paenitz operator on $S^{4}$ described in \cite{beckner} and elsewhere.
As $2\in\sigma(D_{S})$ when $n=4$ it may be seen that $0$ is in the
spectrum of $D_{S}-2w$. Consequently when $n=4$ zero is in the spectrum
of the Paenitz operator and so this operator is not invertible on $L^{2}(S^{4})$.
It is easy to see that it is invertible in all other dimensions.

\section[Some Sharp $L^{2}$ inequalities on $S^{n}$ and $\mathbb{R}^{n}$]{Some Sharp $\boldsymbol{L^{2}}$ inequalities on $\boldsymbol{S^{n}}$ and $\boldsymbol{\mathbb{R}^{n}}$}

\begin{theorem}\label{thm1}
Suppose that $\phi:S^{n}\rightarrow Cl_{n+1}$ is a $C^{1}$ function. Then
\[\|D_{S}\phi\|_{L^{2}}\geq\frac{n}{2}\|\phi\|_{L^{2}}.\]
\end{theorem}

\begin{proof} As $\phi\in C^{1}(S^{n})$ then $\phi\in L^{2}(S^{n})$. It follows that
\[\phi=\sum_{m=0}^{\infty}\sum_{p_{m}\in{\it{P}}_{m}}p_{m}+\sum_{m=0}^{-\infty}
\sum_{q_{m}\in{\it{Q}}_{m}}q_{m},\]
where $p_{m}$ and $q_{m}$ are eigenvectors of $\Gamma_{w}$.
Further the eigenvectors $p_{m}$ can be chosen so that within~${\it{P}}_{m}$
they are mutually orthogonal. The same can be done for the eigenvectors $q_{m}$.
Moreover as $\phi\in C^{1}$ then $D_{S}\phi\in C^{0}(S^{n})$ and so $D_{S}\phi\in L^{2}(S^{n})$.
Consequently
\[D_{S}\phi=w\left(\sum_{m=0}^{\infty}\left(m+\frac{n}{2}\right)\sum_{p_{m}\in{\it{P}}_{m}}p_{m}+
\sum_{m=0}^{\infty}\left(-\frac{n}{2}-m\right)\sum_{q_{m}\in{\it{Q}}_{m}}q_{m}\right).\]
But $wp_{m}(w)\in{\it{Q}}_{m}$ and $wq_{m}(w)\in{\it{P}}_{m}$.
Consequently
\[D_{S}\phi=\sum_{m=0}^{\infty}\left(m+\frac{n}{2}\right)\sum_{q_{m}\in{\it{Q}}_{m}}q_{m}+
\sum_{m=0}^{\infty}\left(-\frac{n}{2}-m\right)\sum_{p_{m}\in{\it{P}}_{m}} p_m.\]
It follows that
\begin{gather*}
\|D_{S}\phi\|_{L^{2}}=\sum_{m=0}^{\infty}\left(m+\frac{n}{2}\right)^{2}
\sum_{q_{m}\in{\it{Q}}_{m}}\|q_{m}\|_{L^{2}}^{2}+\sum_{m=0}^{\infty}\left(-\frac{n}{2}-m\right)^{2}
\sum_{p_{m}\in{\it{P}}_{m}}\|p_{m}\|_{L^{2}}^{2}\\
\phantom{\|D_{S}\phi\|_{L^{2}}=}{} \geq \left(\frac{n}{2}\right)^{2}\left(\sum_{m=0}^{\infty}\sum_{p_{m}\in{\it{P}}_{m}}\|p_{m}\|_{L^{2}}^{2}+
\sum_{m=0}^{-\infty}\sum_{q_{m}\in{\it{Q}}_{m}}\|q_{m}\|_{L^{2}}^{2}\right)
\end{gather*}
as $\pm\frac{n}{2}$ are the smallest eigenvalues of
$\Gamma_{w}+\frac{n}{2}$. That is $\pm\frac{n}{2}$ are the
eigenvalues closest to zero. Therefore
\[\|D_{S}\phi\|_{L^{2}}^{2}\geq \left(\frac{n}{2} \right)^{2}\|\phi\|_{L^{2}}^{2}.\]
The result follows.
\end{proof}

It should be noted from the proof of Theorem~\ref{thm1} that this
inequality is sharp.

In the proof of Theorem~\ref{thm1} it is noted that the operator $D_{S}$
takes ${\it{P}}_{m}$ to ${\it{Q}}_{m}$ and it takes ${\it{Q}}_{m}$
to ${\it{P}}_{m}$. This is also true of the operator $D_{S}+\alpha
w$ for any $\alpha\in\mathbb{C}$. As
$\triangle_{S}=D_{S}(D_{S}+w)$ it now follows that the spectrum,
$\sigma(\triangle_{S})$, of the conformal Laplacian,
$\triangle_{S}$, is $\{-(\frac{n}{2}+m)(\frac{n}{2}+m+1)$,
$-(\frac{n}{2}+m)(\frac{n}{2}+m-1):m\in\mathbb{N}\cup\{0\}\}$. So
the smallest eigenvalue is $\frac{n(2-n)}{4}$. We therefore have
the following sharp inequality:

\begin{theorem}\label{thm2}
Suppose $\phi:S^{n}\rightarrow Cl_{n+1}$ is a $C^{2}$ function. Then
\[\|\triangle_{S}\phi\|_{L^{2}}\geq\frac{n(n-2)}{4}\|\phi\|_{L^{2}}.\]
\end{theorem}

We now proceed to generalize Theorems \ref{thm1} and \ref{thm2} for all operators
$D_{S}^{(k)}$. We begin with:

\begin{lemma}\label{lem2}
(i) For $k$ even the smallest eigenvalue of $D_{S}^{k}$ is
\[\frac{n(2-n)\cdots(n+k-2)(k-n)}{2^{k}}\]
and

(ii) for $k$ odd
\[\frac{n(n+2)(2-n)\cdots(n+k-1)(k-1-n)}{2^{k}}.\]
\end{lemma}

\begin{proof} Let us f\/irst assume that $k$ even. As $D_{S}+\alpha
w:{\it{P}}_{m}\rightarrow {\it{Q}}_{m}$ and $D_{S}+\alpha
w:{\it{Q}}_{m}\rightarrow {\it{P}}_{m}$ for any
$\alpha\in\mathbb{R}$ then
\[\frac{(n+2m)(2-n-2m)\cdots(n+k-2+2m)(k-n-2m)}{2^{k}}\]
 and
\[\frac{(2m-n)(n+2+2m) \cdots(k-2-n-2m)(n+k+2m)}{2^{k}}\]
 are eigenvalues of $D_{S}^{(k)}$ for $m=0,1,2,\ldots$.
But for any positive even integer $l$ the term
$(n+l-2+2m)(l-n-2m)$ is closer to zero than $(l-2-n-2m)(n+l+2m)$.
The result follows for~$k$ even. The case $k$ is odd is proved
similarly.
\end{proof}

It should be noted that when $n$ is even and $k\geq n$ then $0$ is
an eigenvalue of $D_{S}^{(k)}$. Consequently in these cases
$D_{S}^{(k)}$ is not an invertible operator on $L^{2}(S^{n})$.

From Lemma \ref{lem2} we have:

\begin{theorem}\label{thm3}
Suppose $\phi:S^{n}\rightarrow Cl_{n+1}$ is a $C^{k}$ function. Then for $k$ even
\[\|D_{S}^{(k)}\phi\|_{L^{2}}\geq \frac{|n(2-n)\cdots(n+k-2)(k-n)|}{2^{k}}\|\phi\|_{L^{2}}\]
and for $k$ odd
\[\|D_{S}^{(k)}\phi\|_{L^{2}}\geq
\frac{|n(n+2)(2-n)\cdots(n+k-1)(k-1-n)|}{2^{k}}\|\phi\|_{L^{2}}.\]
\end{theorem}

Again these inequalities are sharp.

When $n$ is odd then of course $\frac{n}{2}$ is not an integer. It
follows that in odd dimensions zero is not an eigenvalue for the
operator $D_{S}^{(k)}$. In the cases $n$ even and $k\geq n$ the
smallest eigenvalue is zero so for those cases the inequality in
Theorem~\ref{thm3} is trivial. This includes the Paenitz operator on~$S^{4}$. It follows that none of these operators have fundamental
solutions. The fundamental solutions for $D_{S}^{(k)}$ for all $k$
when $n$ is odd and for $1\leq k<n$ when $n$ is even are given in~\cite{lr}. We shall denote them by $C_{k}(w,y)$.

As $\frac{n(n+2)(2-n)\cdots(n+k-1)(nk-1-n)}{2^{k}}$ is the
smallest eigenvalue for $D_{S}^{(k)}$ for $k$ odd then
\[\frac{-2^{k}}{n(n+2)(2-n)\cdots(n+k-1)(k-1+n)}\]
is the largest eigenvalue of $D_{S}^{(k)-1}$ for $n$ odd or for $1\leq k\leq n-1$ when $n$ is even.

Similarly for $k$ even and $n$ odd and $k$ even with $1<k<n-1$ for
$n$ even
\[\frac{2^{k}}{n(2-n)\cdots(n+k-2)(k-n)}\]
is the largest eigenvalue of $D_{S}^{(k)-1}$.

Similarly to Theorem \ref{thm3} we now have the following sharp inequality:

\begin{theorem}\label{thm4}
Suppose $\phi:S^{n}\rightarrow Cl_{n+1}$ is a continuous function. Then for $n$ odd and $k$ even and for $n$ even and $k$ even with $1<k<n$
\[\|C_{k}(w,y)\star\phi(w)\|_{L^{2}}\leq\frac{2^{k}}{|n(2-n)\cdots(n+k-2)(k-n)|}\|\phi\|_{L^{2}}\]
and for $n$ odd and $k$ odd and $n$ even and $k$ odd with $1\leq k\leq n-1$
\[\|C_{k}(w,y)\star\phi(w)\|_{L^{2}}\leq\frac{2^{k}}{|n(n+2)(2-n))\cdots(n+k-1)(k-1-n)|}\|\phi\|_{L^{2}}.\]
\end{theorem}

Let us now turn to $\mathbb{R}^{n}$ and retranslate Theorems \ref{thm3} and
\ref{thm4} in this context. In \cite{lr} the Cayley transformation
$C(x)=(e_{n+1}+1)(x+e_{n+1})^{-1}$ is used to show that
\begin{equation}\label{eq1}
D_{S}^{(k)}=J_{-k}(C,x)^{-1}D^{k}J_{k}(C,x),
\end{equation}
where
$J_{k}(C,x)=\frac{2^{\frac{n-k}{2}}(x+e_{n+1})}{(\|1+\|x\|^{2})^{\frac{n-k+1}{2}}}$
when $k$ is odd and
$J_{k}(C,x)=\frac{2^{\frac{n-k}{2}}}{(1+\|x\|^{2})^{\frac{n-k}{2}}}$
when $k$ is even. Note that $J_{k}(C,x)\in GPin(n+1)$. By applying
Lemma 1 we now see that on $\mathbb{R}^{n}$ the Cayley
transformation can be applied to Theorem \ref{thm3} to give:

\begin{theorem}\label{thm5}
Suppose $\phi:\mathbb{R}^{n}\rightarrow Cl_{n+1}$ is a $C^{k}$ function
with compact support. Then for each $k\in\mathbb{N}$ for $n$ odd and for
$k=1,\ldots, n-1$ for $n$ even
\begin{gather*}
\left(\int_{\mathbb{R}^{n}}\|D^{k}\phi(x)\|^{2}(1+\|x\|^{2})^{k}dx^{n}\right)^{\frac{1}{2}}\\
\qquad{} \geq
 |n(n+2)\cdots(n+k-1)(k-1-n)|
\left(\int_{\mathbb{R}^{n}}\frac{\|\phi\|^{2}2^{k}}{(1+\|x\|^{2})^{k}}dx^{n}\right)^{\frac{1}{2}}
\end{gather*}
for $k$ odd, and
\begin{gather*}
\left(\int_{\mathbb{R}^{n}}\|\triangle_{n}^{\frac{k}{2}}\phi(x)\|^{2}(1+\|x\|^{2})^{k} dx^{n}
\right)^{\frac{1}{2}}\\
\qquad {}\geq |n(2-n)\cdots(n+k-2)(k-n)|\left(\int_{\mathbb{R}^{n}}
\frac{\|\phi(x)\|^{2}2{k}}{(1+\|x\|^{2})^{k}}dx^{n}\right)^{\frac{1}{2}}
\end{gather*}
for $k$ even.
\end{theorem}

\begin{proof} For any M\"{o}bius transformation $M(x)=(ax+b)(cx+d)^{-1}$ the
associated Jacobian over a domain in $\mathbb{R}^{n}$ is $\frac{2^{n}}{\|cx+d\|^{2n}}$.
Consequently for $\psi:S^{n}\rightarrow Cl_{n+1}$ a $C^{k}$ function
the integral $\int_{S^{n}}\| D_{S}^{(k)}\psi(w)\|^{2}d\sigma(w)$ by equation \eqref{eq1} becomes
\[\int_{\mathbb{R}^{n}}\| J_{-k}(C,x)^{-1}D^{k}J_{k}(C,x)\psi(C(x))\|^{2}\frac{2^{n}dx^{n}}{(1+\|x\|^{2})^{n}}.\]
By Lemma \ref{lem1} this expression becomes
\[\frac{1}{2^{k}}\int_{\mathbb{R}^{n}}(1+\|x\|^{2})^{k}\| D^{k}J_{k}(C,x)\psi(C(x))\|^{2}dx^{n}.\]
Further
\begin{gather*}
\int_{S^{n}}\|\psi\|^{2}d\sigma(x)=\int_{\mathbb{R}^{n}}\|\psi(C(x))\|^{2}\frac{2^{n}dx^{n}}{(1+\|x\|^{2})^{n}}
\\ \phantom{\int_{S^{n}}\|\psi\|^{2}d\sigma(x)}{}
=\int_{\mathbb{R}^{n}}\|J_{k}(C,x)^{-1}J_{k}(C,x)\psi(C(x))\|\frac{2^{n}dx^{n}}{(1+\|x\|^{2})^{n}}.
\end{gather*}
By Lemma \ref{lem1} this last expression becomes
\[2^{k}\int_{\mathbb{R}^{n}}\|J_{k}(C,x)\psi(x)\|^{2}(1+\|x\|^{2})^{-k}dx^{n}.\]
On placing $J_{k}(C,x)\psi(C(x))=\phi(x)$ Theorem \ref{thm3} now gives the
result.
\end{proof}

In \cite{lr} it is shown that the kernel $C_{k}(w,y)$ is
conformally equivalent to the kernel $G_{k}(x-y)$ in
$\mathbb{R}^{n}$, where
$G_{k}(x-y)=\frac{C_{k}}{\omega_{n}}\frac{x-y}{\|x-y\|^{n+1-k}}$
when $k$ is odd and
$G_{k}(x-y)=\frac{C_{k}}{\omega_{n}}\frac{1}{\|x-y\|^{n-k}}$ when
$k$ is even. Here $C_{k}$ is a real constant chosen so that
$DG_{k}=G_{k-1}$ for $k>1$ and with $C_{1}=1$.

As $J_{-k}(C,x)^{-1}D_{S}^{(k)}J_{k}(C,x)=D^{k}$ then
$D^{-k}=J_{k}(C,x)^{-1}D_{S}^{(k)-1}J_{-k}(C,x)$. Consequently:

\begin{theorem}\label{thm6}
Suppose $h:\mathbb{R}^{n}\rightarrow Cl_{n+1}$ is a continuous function with compact support.
Then for $n$ odd and $k$ odd and for $n$ even and any odd integer $k$ satisfying $1\leq k<n$
\begin{gather*}
\left(\int_{\mathbb{R}^{n}}\left\|\int_{\mathbb{R}^{n}}G_{k}(x-y) h(x)dx^{n})\right\|^{2}
\frac{1}{(1+\|y\|^{2})^{k}}dy^{n}\right)^{\frac{1}{2}}\\
\qquad {}\leq\frac{1}{|n(n+2)\cdots(n+k-1)(k-1-n)|}
\left(\int_{\mathbb{R}^{n}}\|h(x)\|^{2}(1+\|x\|^{2})^{k}dx^{n}\right)^{\frac{1}{2}}
\end{gather*}
and for $n$ odd and $k$ even and for $n$ even and $k$ an even integer satisfying $1<k<n$
\begin{gather*}
\left(\int_{\mathbb{R}^{n}}\left\|\int_{\mathbb{R}^{n}}G_{k}(x-y) h(x)dx^{n}\right\|^{2}
\frac{1}{(1+\|y\|^{2})^{k}}dy^{n}\right)^{\frac{1}{2}}\\
\qquad{} \leq\frac{1}{|n(n+2)(2-n)\cdots(n+k-2)(k-n)|}
\left(\int_{\mathbb{R}^{n}}\|h(x)\|^{2}(1+\|x\|^{2})^{k}dx^{n}\right)^{\frac{1}{2}}.
\end{gather*}
\end{theorem}

\section[Dirac type operators in $\mathbb{R}^{n}$]{Dirac type operators in $\boldsymbol{\mathbb{R}^{n}}$}

In this section we demonstrate a somewhat alternative approach to
obtained Theorems \ref{thm5} and \ref{thm6}.

We have previously seen that $D_{S}p_{m}=(m+\frac{n}{2})p_{m}$ for
$p_{m}\in{\it{P}}_{m}$, that $D_{S}q_{m}=(-\frac{n}{2}-m)q_{m}$
for $q_{m}\in{\it{Q}}_{m}$ and
$J_{-1}^{-1}(C,x)DJ_{1}(C,x)=D_{S}$. Consequently
\[DJ_{1}(C,x)p_{m}(C(x))=\frac{2}{1+\|x\|^{2}}\left(m+\frac{n}{2}\right)J_{1}(C,x)p_{m}(C(x))\]
and
\[DJ_{1}(C,x)q_{m}(C(x))=\frac{2}{1+\|x\|^{2}}\left(-\frac{n}{2}-m\right)J_{1}(C,x)q_{m}(C(x)).\]

Further:
\begin{proposition}\label{prop1}
$\psi(w)\in L^{2}(S^{n})$ if and only if $\frac{1}{(1+\|x\|^{2})^{\frac{1}{2}}}J_{1}(C,x)\psi(C(x))\in L^{2}(\mathbb{R}^{n})$. Further if $\psi'(x)=\frac{1}{(1+\|x\|^{2})^{\frac{1}{2}}}J_{1}(C,x)\psi(C(x))$ and $\phi'(x)=\frac{1}{(1+\|x\|^{2})^{\frac{1}{2}}}J_{1}(C,x)\phi(C(x))$ for $\psi$ and $\phi\in L^{2}(S^{n})$ then
\[\int_{S^{n}}\overline{\phi}(w)\psi(w)d\sigma(w)=\int_{\mathbb{R}^{n}}\overline{\phi}'(x)\psi'(x)dx^{n}.\]
\end{proposition}

This leads us to:
\begin{theorem}\label{thm7}
Suppose $h:\mathbb{R}^{n}\rightarrow Cl_{n+1}$ is a smooth function with compact support. Then
\[\left(\int_{\mathbb{R}^{n}}\|Dh(x)\|^{2}(1+\|x\|^{2})dx^{n}\right)^{\frac{1}{2}}\geq
n\left(\int_{\mathbb{R}^{n}}\frac{\|h(x)\|^{2}}{1+\|x\|^{2}}dx^{n}\right)^{\frac{1}{2}}.\]
\end{theorem}

In \cite{lr} it is shown that
$J_{-2}(C,x)^{-1}\triangle_{S}J_{2}(C,x)=\triangle_{n}$.
Proposition \ref{prop1} can easily be adapted replacing $J_{1}(C,x)$ by
$J_{2}(C,x)$ and $\frac{2}{1+\|x\|^{2}}$ by
$\frac{4}{(1+\|x\|^{2})^{2}}$. From Theorem \ref{thm2} we now have:

\begin{theorem}\label{thm8}
Suppose that $h$ is as in Theorem~{\rm \ref{thm7}}. Then
\[\left(\int_{\mathbb{R}^{n}}\|\triangle_{n} h(x)\|^{2}(1+\|x\|^{2})^{2}dx^{n}\right)^{\frac{1}{2}}
\geq
n(n-2)\left(\int_{\mathbb{R}^{n}}\frac{\|h(x)\|^{2}}{(1+\|x\|^{2})^{2}}dx^{n}\right)^{\frac{1}{2}}.\]
\end{theorem}

Using Lemma \ref{lem2} we also have

\begin{theorem}\label{thm9}
Suppose $h$ is as in Theorem~{\rm \ref{thm7}}. Then for $n$ odd and $k$ even and for $n$ even and $k$ an even
integer belonging to $\{1,\ldots, n-1\}$
\begin{gather*}
\left(\int_{\mathbb{R}^{n}}\|\triangle_{n}^{\frac{k}{2}}h(x)\|^{2}(1+\|x\|^{2})^{k}dx^{n}\right)^{\frac{1}{2}}\\
\qquad {} \geq
|n(2-n)\cdots(n+k-2)(k-n)|\left(\int_{\mathbb{R}^{n}}\frac{\|h(x)\|^{2}}{(1+\|x\|^{2})^{k}}dx^{n}\right)^{\frac{1}{2}}
\end{gather*}
and for $n$ odd and $k$ odd and for $n$ even and $k$ belonging to
$\{1,\ldots,n-1\}$
\begin{gather*}
\left(\int_{\mathbb{R}^{n}}\|D^{k}h(x)\|^{2}(1+\|x\|^{2})^{k})dx^{n}\right)^{\frac{1}{2}}\\
\qquad {}\geq
 |n(n+2)(2-n)\cdots(n+k-1)(k-1-n)|
\left(\int_{\mathbb{R}^{n}}\frac{\|h(x)\|^{2}}{(1+\|x\|^{2})^{k}}dx^{n}\right)^{\frac{1}{2}}.
\end{gather*}
\end{theorem}

\begin{theorem}\label{thm10}
Suppose $h:\mathbb{R}^{n}\rightarrow Cl_{n+1}$ is a continuous function with compact support.
Then for $k$ odd and $n$ odd and for $n$ even and $k$ odd and satisfying $1\leq k\leq n-1$
\begin{gather*}
\left(\int_{\mathbb{R}^{n}}\frac{\|G_{k}\star h(x)\|^{2}}{(1+\|x\|^{2})^{k}}dx^{n}\right)^{\frac{1}{2}}\\
\qquad {} \leq
 \frac{1}{|n(n+2)(2-n)\cdots(n+k-1)(k-1-n)|}
\left(\int_{\mathbb{R}^{n}}\|h(x)\|^{2}(1+\|x\|^{2})^{k}dx^{n}\right)^{\frac{1}{2}}
\end{gather*}
and for $n$ odd and $k$ even and for $n$ even and $k$ even and
satisfying $1<k<n$
\begin{gather*}
\left(\int_{\mathbb{R}^{n}}\frac{\|G_{k}\star h(x)\|^{2}}{(1+\|x\|^{2})^{k}}dx^{n}\right)^{\frac{1}{2}}\\
\qquad {} \leq
 \frac{1}{|n(n+2)(2-n)\cdots(n+k-2)(k-n)|}
\left(\int_{\mathbb{R}^{n}}\|h(x)\|^{2}(1+\|x\|^{2})^{k}dx^{n}\right)^{\frac{1}{2}}.
\end{gather*}
\end{theorem}

\section{Concluding remarks}

Let us consider the Paenitz operator on $S^{5}$. Via the Cayley
transform this operator stereographically projects to the
bi-Laplacian, $\triangle_{5}^{2}$ on $\mathbb{R}^{5}$. If we
restrict attention to the equator, $S^{4}$, of  $S^{5}$ we see
that the restriction of the Paenitz operator in this context
stereographically projects to the restriction of
$\triangle_{5}^{2}$ to $R^{4}$. This operator is the bi-Laplacian
$\triangle_{4}^{2}$ in $\mathbb{R}^{4}$, while the restriction of
the Paenitz operator on $S^{5}$ to its equator, $S^{4}$, is the
Paenitz operator on $S^{4}$. The Paenitz operator on $S^{4}$ has a
zero eigenvalue. Consequently there is no real hope of obtaining
inequalities of the type we have obtained here in $\mathbb{R}^{n}$
for the bi-Laplacian in $\mathbb{R}^{4}$. This should explain the
breakdown of the Rellich inequality, described in \cite{dh},  for
the bi-Laplacian in $\mathbb{R}^{4}$. The same rationale also
explains similar breakdowns of inequalities for $D^{k}$ in
$\mathbb{R}^{n}$ for $n$ even and $k\geq n$.

It should be clear that similar sharp $L^{2}$ inequalities can be
obtained for the operator $D_{S}+\alpha w$ provided $-\alpha$ is
not in the spectrum of $wD_{S}$. These operators conformally
transform to $D+\frac{\alpha}{1+\|x\|^{2}}$ in $\mathbb{R}^{n}$.
When $-\alpha$ is in the spectrum of $wD_{S}$ then we obtain a
f\/inite dimensional subspace of the weighted $L^{2}$ space
$L^{2}(\mathbb{R}^{n},(1+\|x\|^{2})^{-2})$, with weight
$(1+\|x\|^{2})^{-2}$, consisting of solutions to the Dirac
equation $Du+\frac{\alpha}{1+\|x\|^{2}} u=0$.

All inequalities obtained here are $L^{2}$ inequalities. It would
be nice to see similar inequalities for other suitable $L^{p}$
spaces.


\subsection*{Acknowledgements} The authors are grateful to the Royal Society for
support of this work under grant 2007/R1.

\pdfbookmark[1]{References}{ref}
\LastPageEnding
\end{document}